\documentclass{llncs}

\usepackage{fullpage}
\usepackage{epsfig}
\usepackage{latexsym}
\usepackage{amsmath}
\usepackage{amsfonts}
\usepackage{amssymb}
\usepackage{graphicx}
\usepackage{algorithm}
\usepackage{myalgorithmic}
\usepackage{enumitem}
\usepackage{changepage}
\usepackage{booktabs}
\usepackage{multirow}
\usepackage{array}
\usepackage{tikz}
\usepackage{float}
\usepackage{sectsty}
\usepackage[mathscr]{euscript}
\usepackage{color}

\chapterfont{\raggedright}

\newfloat{algorithm}{t}{lop}
\newcolumntype{C}[1]{>{\centering\arraybackslash}p{#1}}

\usetikzlibrary{arrows,shapes,snakes,automata,backgrounds,petri}
\usetikzlibrary{positioning,calc}
\usetikzlibrary{patterns}
\tikzstyle{legendborder}=[rectangle, draw, black, rounded corners, thin, top color=white, text=black, minimum width=2.5cm, text width=4.5cm]
\tikzstyle{legendnoborder}=[rectangle, draw, white, rounded corners, thin, top color=white, text=black, minimum width=2.5cm]
\tikzstyle{selected edge} = [draw,line width=1pt,black]
\usetikzlibrary{shapes,backgrounds,plothandlers,plotmarks,calc,arrows,fadings,decorations.pathreplacing,decorations.pathmorphing,}

\tikzstyle{v}=[circle,fill=black,draw=black!75,inner sep=0pt,minimum size=0.3em]
\tikzstyle{I}=[circle,draw=black!75,inner sep=0pt,minimum size=0.8em]
\tikzstyle{J}=[rectangle,draw=black!75,inner sep=0pt,minimum size=0.7em]

\newcommand{\mc}{\mathcal}
\newcommand{\Oh}{\mc{O}}

\newcommand{\PP}{\textnormal{\texttt{\textbf{P}}}}
\newcommand{\NP}{\textnormal{\texttt{\textbf{NP}}}}

\newcommand{\PSPACE}{\textnormal{\texttt{\textbf{PSPACE}}}}

\newcommand{\FPT}{\textnormal{\texttt{\textbf{FPT}}}}
\newcommand{\WW}{\textnormal{\texttt{\textbf{W}}}}
\newcommand{\WONE}{\textnormal{\texttt{\textbf{W[1]}}}}
\newcommand{\WTWO}{\textnormal{\texttt{\textbf{W[2]}}}}

\newcommand{\WT}{\textnormal{\texttt{\textbf{W[t]}}}}
\newcommand{\WTT}{\textnormal{\texttt{\textbf{W[t + 1]}}}}
\newcommand{\XP}{\textnormal{\texttt{\textbf{XP}}}}

\renewcommand\algorithmicthen{}
\renewcommand\algorithmicdo{}

\newcommand{\sousol}{S_s}
\newcommand{\tarsol}{S_t}

\makeatletter
\newcommand\xleftrightarrow[2][]{%
  \ext@arrow 9999{\longleftrightarrowfill@}{#1}{#2}}
\newcommand\longleftrightarrowfill@{%
  \arrowfill@\leftarrow\relbar\rightarrow}
\makeatother

\begin{document}
\title{Reconfiguration on sparse graphs}
\author
{
    Daniel Lokshtanov\inst{1}
    \and Amer E. Mouawad\inst{2}
    \and Fahad Panolan\inst{3}
    \and M.S. Ramanujan\inst{1}
    \and Saket Saurabh\inst{3}
}
\institute
{
    University of Bergen, Norway.\\
    \email{daniello,ramanujan.sridharan@ii.uib.no}\\
    \and
    David R. Cheriton School of Computer Science\\
    University of Waterloo, Ontario, Canada.\\
    \email{aabdomou@uwaterloo.ca}
    \and
    Institute of Mathematical Sciences\\
    Chennai, India.\\
    \email{fahad,saket@imsc.res.in}\\
}
\maketitle
\sloppy

\begin{abstract}
A vertex-subset graph problem $\mathcal{Q}$ defines which subsets of
the vertices of an input graph are feasible solutions. A
reconfiguration variant of a vertex-subset problem asks, given
two feasible solutions $\sousol$ and $\tarsol$ of size $k$,
whether it is possible to transform $\sousol$ into $\tarsol$
by a sequence of vertex additions and deletions such that each intermediate
set is also a feasible solution of size bounded by $k$.
We study reconfiguration variants of two classical vertex-subset problems,
namely {\sc Independent Set} and {\sc Dominating Set}.
We denote the former by $\textsc{ISR}$ and the latter by $\textsc{DSR}$.
Both $\textsc{ISR}$ and $\textsc{DSR}$ are \PSPACE-complete on graphs of bounded bandwidth
and \WONE-hard parameterized by $k$ on general graphs.
We show that $\textsc{ISR}$ is fixed-parameter tractable
parameterized by $k$ when the input
graph is of bounded degeneracy or nowhere-dense.
As a corollary, we answer positively an open question concerning the parameterized complexity
of the problem on graphs of bounded treewidth.
Moreover, our techniques generalize recent results showing that $\textsc{ISR}$
is fixed-parameter tractable on planar graphs and graphs of bounded degree.
For $\textsc{DSR}$, we show the problem fixed-parameter
tractable parameterized by $k$ when the input graph does not
contain large bicliques, a class of graphs
which includes graphs of bounded degeneracy and nowhere-dense graphs.
\end{abstract}

\section{Introduction}
Given an $n$-vertex graph $G$ and two vertices $s$ and $t$ in $G$,
determining whether there exists a path and computing the length of the shortest
path between $s$ and $t$ are two of the most fundamental graph problems.
In the classical battle of \PP\ versus \NP\ or ``easy'' versus ``hard'', both
of these problems are on the easy side. That is, they can be solved in $poly(n)$ time,
where $poly$ is any polynomial function.
But what if our input consisted of a $2^n$-vertex graph?
Of course, we can no longer assume $G$ to be part of the input, as reading the input
alone requires more than $poly(n)$ time. Instead, we are given an oracle
encoded using $poly(n)$ bits and
that can, in constant or $poly(n)$ time, answer queries
of the form ``is $u$ a vertex in $G$'' or ``is there an edge between $u$ and $v$?''.
Given such an oracle and two vertices of the $2^n$-vertex graph,
can we still determine if there is a path or compute the length of
the shortest path between $s$ and $t$ in $poly(n)$ time?

A slightly different, but equally insightful, formulation of the question above is as follows.
Given a set $S$ of $n$ objects, consider the graph $R(S)$ which
contains one node for each set in the power set of $S$, $2^S$, and two nodes
are adjacent in $R(S)$ whenever the size of their symmetric difference
is equal to one. Clearly, this graph contains $2^n$ nodes and can
be easily encoded in $poly(n)$ bits using the oracle described above.
It is not hard to see that there exists a path between
any two nodes of $R(S)$. Moreover,
computing the length of a shortest path can be accomplished in constant time;
it is equal to the size of the symmetric difference of the two underlying sets.
If the node set of $R(S)$ were instead restricted to a subset of $2^S$,
both of our problems can become \NP-complete or even \PSPACE-complete.
Therefore, another interesting question is whether we can determine
what types of ``restriction'' on the node set of $R(S)$
induce such variations in the complexity of the two problems.

These two seemingly artificial questions are in fact quite natural and
appear in many practical and theoretical problems.
In particular, these are exactly the types of questions
asked under the reconfiguration framework, the main subject of this work.
Under the reconfiguration framework, instead of finding a feasible solution to some instance
$\mathcal{I}$ of a search problem $\mathcal{Q}$, we are interested in structural
and algorithmic questions related to the solution space of $\mathcal{Q}$.
Naturally, given some adjacency relation $\mathcal{A}$ defined
over feasible solutions of $\mathcal{Q}$, size of the symmetric difference being one such relation,
the solution space can be represented using a graph $R_{\mathcal{Q}}(\mathcal{I})$.
$R_{\mathcal{Q}}(\mathcal{I})$ contains one node for each feasible solution
of $\mathcal{Q}$ on instance $\mathcal{I}$ and two nodes share an edge whenever their corresponding
solutions are adjacent under $\mathcal{A}$. An edge in $R_{\mathcal{Q}}(\mathcal{I})$ corresponds
to a {\em reconfiguration step}, a walk in $R_{\mathcal{Q}}(\mathcal{I})$ is a sequence of
such steps, a {\em reconfiguration sequence}, and $R_{\mathcal{Q}}(\mathcal{I})$ is
a {\em reconfiguration graph}.

Studying problems related to reconfiguration graphs
has received considerable attention in
recent literature~\cite{B12,GKMP09,IDHPSUU11,IKD12,KMM11,MNRSS13}, the most popular problem being to
determine whether there exists a reconfiguration sequence
between two given feasible solution.
In most cases, this problem was shown~\PSPACE-hard in general, although
some polynomial-time solvable restricted cases have been identified.
For \PSPACE-hard cases, it is not surprising that shortest
paths between solutions can have exponential length.
More surprising is that for most known polynomial-time solvable cases the
diameter of the reconfiguration graph has been shown to be polynomial.
Some of the problems that have been studied under the reconfiguration framework
include {\sc Independent Set}~\cite{KMM12}, {\sc Vertex Cover}~\cite{MNR14},
{\sc Shortest Path}~\cite{Bon12FSTTCS,KMM11}, {\sc Coloring}~\cite{BB13,BC09,BMNR14,Cer07,CHJ09,CHJ11,JKKPP14},
and {\sc Boolean Satisfiability}~\cite{GKMP09}.
We refer the reader to the recent
survey by Van den Heuvel~\cite{H13} for a detailed overview.
Recently, a systematic study of the parameterized complexity of
reconfiguration problems was initiated by Mouawad et al.~\cite{MNRSS13};
various problems were identified where
the problem was not only \NP-hard (or \PSPACE-hard), but also \WW-hard under
various parameterizations.

\paragraph{Overview of our results.}
In this work, we focus on reconfiguration variants of the {\sc Independent Set (IS)}
and {\sc Dominating Set (DS)} problems.
Given two independent sets $I_s$ and $I_t$ of a graph $G$ such that
$|I_s| = |I_t| = k$, the {\sc Independent Set Reconfiguration} problem asks whether there exists
a sequence of independents sets $\sigma = \langle I_0, I_1, \ldots, I_\ell \rangle$, for some $\ell$, such that:
\begin{itemize}
\item[(1)] $I_0 = I_s$ and $I_\ell = I_t$,
\item[(2)] $I_i$ is an independent set of $G$ for all $0 \leq i \leq \ell$,
\item[(3)] $|I_i \Delta I_{i+1}| = 1$ for all $0 \leq i < \ell$, and
\item[(4)] $k - 1 \leq |S_i| \leq k$ for all $0 \leq i \leq \ell$.
\end{itemize}
%We write $I_s \leftrightarrow I_t$ when such a sequence exists.
Alternatively, given a graph $G$ and integer $k$, the reconfiguration
graph $R_{{\textsc{is}}}(G,k-1,k)$ has a node for each independent set of $G$
of size $k$ or $k - 1$ and two nodes are adjacent in $R_{{\textsc{is}}}(G,k - 1,k)$
whenever the corresponding independent sets can be obtained from one another
by either the addition or the deletion of a single vertex.
The reconfiguration graph $R_{{\textsc{ds}}}(G,k,k+1)$ is defined
similarly for dominating sets.
Hence, {\sc ISR} and {\sc DSR} can be formally stated as follows:

\vspace{8pt}
\noindent
\begin{tabular}{ll}
\multicolumn{2}{l}{{\sc Independent Set Reconfiguration (ISR)}}\\
{\bf Input}: & Graph $G$, positive integer $k$, and two $k$-independent sets $I_s$ and $I_t$\\
{\bf Question}: & Is there a path from $I_s$ to $I_t$ in $R_{{\textsc{is}}}(G, k-1, k)$?
\end{tabular}
\vspace{8pt}

\vspace{8pt}
\noindent
\begin{tabular}{ll}
\multicolumn{2}{l}{{\sc Dominating Set Reconfiguration (DSR)}}\\
{\bf Input}: & Graph $G$, positive integer $k$, and two $k$-dominating sets $D_s$ and $D_t$\\
{\bf Question}: & Is there a path from $D_s$ to $D_t$ in $R_{{\textsc{ds}}}(G, k, k+1)$?
\end{tabular}
\vspace{8pt}

Note that since we only allow independent sets of size $k$ and $k - 1$
the ISR problem is equivalent to reconfiguration under the token jumping model considered by Ito et al.~\cite{IKO14,IKOSUY14}.
{\sc ISR} is known to be \PSPACE-complete on graphs of
bounded bandwidth~\cite{MNRW14,WROCHNA14} (hence pathwidth and treewidth)
and \WONE-hard on general graphs~\cite{IKOSUY14}. On the positive side, the problem was shown
fixed-parameter tractable, with parameter $k$, for graphs of
bounded degree, planar graphs, and graphs excluding
$K_{3,d}$ as a (not necessarily induced) subgraph, for any constant $d$~\cite{IKO14,IKOSUY14}.
We push this boundary further by showing that the problem remains
fixed-parameter tractable for graphs of
bounded degeneracy and nowhere-dense graphs (Figure~\ref{fig-graph-classes}).
As a corollary, we answer positively an open question concerning the parameterized complexity
of the problem (parameterized by $k$) on graphs of bounded treewidth.

For {\sc DSR}, we first show that the problem is \WONE-hard on general graphs
by adapting the well-known (parameter-preserving) reduction from {\sc Independent Set} to {\sc Dominating Set}.
Then, we show that the problem is fixed-parameter tractable, with parameter $k$,
for graphs excluding $K_{d,d}$ as a (not necessarily induced) subgraph, for any constant $d$.
Note that this class of graphs includes both nowhere-dense and bounded degeneracy graphs and
is the ``largest'' class on which the {\sc Dominating Set} problem is known to be in \FPT~\cite{PRS09,TV12}.

Clearly, our main open question is whether {\sc ISR} remains fixed-parameter
tractable on graphs excluding $K_{d,d}$ as a subgraph. Intuitively,
all of the classes we consider fall under the category of ``sparse'' graph classes.
Hence, in some sense, one would not expect a sparse graph to have ``too many'' dominating
sets of fixed small size $k$ as $n$ becomes larger and larger. For independent sets,
the situation is reversed. As $n$ grows larger, so does the number of independent sets
of fixed size $k$. So it remains to be seen whether
some structural properties of graphs excluding $K_{d,d}$ as a subgraph can be used
to settle our open question or whether the problem becomes \WONE-hard.
In the latter case, this would be the first example of a \WONE-hard problem (in general),
which is in \FPT\ on a class $\mathscr{C}$ of graphs but where the reconfiguration version is not;
finding such a problem, we believe, is interesting in its own right.
Another open question is whether we can adapt our results for {\sc ISR}
to find shortest reconfiguration sequences. Our algorithm for {\sc DSR} does in fact
guarantee shortest reconfiguration sequences but, as we shall
see, the same does not hold for both {\sc ISR} algorithms.

\section{Preliminaries}
For an in-depth review of general graph
theoretic definitions we refer the reader to the book of Diestel~\cite{D05}.
Unless otherwise stated, we assume that each graph $G$ is a
simple, undirected graph with vertex set $V(G)$ and
edge set $E(G)$, where $|V(G)| = n$ and $|E(G)| = m$.
The {\em open neighborhood}, or simply {\em neighborhood}, of a
vertex $v$ is denoted by $N_G(v) = \{u \mid uv \in E(G)\}$, the
{\em closed neighborhood} by $N_G[v] = N_G(v) \cup \{v\}$. Similarly, for a set of vertices $S \subseteq V(G)$,
we define $N_G(S) = \{v \mid uv \in E(G), u \in S, v \not\in S \}$ and $N_G[S] = N_G(S) \cup S$.
The {\em degree} of a vertex is $|N_G(v)|$.
We drop the subscript $G$ when clear from context.
A {\em subgraph} of $G$ is a graph $G'$
such that $V(G') \subseteq V(G)$ and $E(G') \subseteq E(G)$.
The {\em induced subgraph} of $G$ with respect to $S \subseteq V(G)$ is denoted by $G[S]$;
$G[S]$ has vertex set $S$ and edge set $\{uv \in E(G[S]) \mid u, v \in S, uv \in E(G)\}$.
We denote by $\Delta(G)$ and $\delta(G)$ the maximum and minimum degree of $G$, respectively.

A {\em walk} of length $\ell$ from $v_0$ to $v_\ell$ in $G$ is a vertex sequence $v_0, \ldots, v_\ell$, such that
for all $i \in \{0, \ldots, \ell-1\}$, $v_iv_{i + 1} \in E(G)$.
It is a {\em path} if all vertices are distinct. It is a {\em cycle}
if $\ell \geq 3$, $v_0 = v_\ell$, and $v_0, \ldots, v_{\ell - 1}$ is a path.
A path from vertex $u$ to vertex $v$ is also called a {\em $uv$-path}.
The {\em distance} between two vertices $u$
and $v$ of $G$, $dist_G(u, v)$, is the length of a shortest $uv$-path in $G$ (positive infinity if no such path exists).
The {\em eccentricity} of a vertex $v \in V(G)$, $ecc(v)$, is equal to $max_{u \in V(G)}(dist_G(u,v))$.
The {\em radius} of $G$, $rad(G)$, is equal to $min_{v \in V(G)}(ecc(v))$.
The {\em diameter} of $G$, $diam(G)$, is equal to $max_{v \in V(G)}(ecc(v))$.
For $r \geq 0$, the {\em $r$-neighborhood} of a vertex $v \in V(G)$ is defined as
$N^r_G[v] = \{u \mid dist_G(u, v) \leq r\}$. We write
$B(v, r) = N^r_G[v]$ and call it a {\em ball of radius $r$ around $v$};
for $S \subseteq V(G)$, $B(S, r) = \bigcup_{v \in S} N^r_G[v]$.

{\em Contracting} an edge $uv$ of $G$ results in a new graph $H$ in which
the vertices $u$ and $v$ are deleted and replaced by a new vertex $w$ that is adjacent
to $N_G(u) \cup N_G(v) \setminus \{u, v\}$.
If a graph $H$ can be obtained from $G$ by repeatedly
contracting edges, $H$ is said to be a {\em contraction} of $G$.
If $H$ is a subgraph of a contraction of $G$, then $H$ is said
to be a {\em minor} of $G$, denoted by $H \preceq_m G$.
An equivalent characterization of minors states that $H$
is a minor of $G$ if there is a map that associates to each vertex $v$ of $H$ a non-empty connected
subgraph $G_v$ of $G$ such that $G_u$ and $G_v$ are disjoint for $u \neq v$ and whenever there is an edge
between $u$ and $v$ in $H$ there is an edge in $G$ between some node in $G_u$ and some node in $G_v$.
The subgraphs $G_v$ are called {\em branch sets}.
$H$ is a {\em minor at depth $r$ of $G$}, $H \preceq^r_m G$, if $H$ is a minor of $G$ which is
witnessed by a collection of branch sets $\{G_v \mid v \in V(H)\}$, each of which induces a graph
of radius at most $r$. That is, for each $v \in V(H)$, there is a $w \in V(G_v)$
such that $V(G_v) \subseteq N^r_{G_v}[w]$.

\paragraph{Sparse graph classes.}
We define the three main classes we consider. Figure~\ref{fig-graph-classes} illustrates
the relationship between these classes and some other well-known classes of sparse graphs.
We refer the reader to~\cite{CLASSES,NO10,NO08} for more details.

\begin{definition}[\cite{NO10,NO08}]\label{def:nowhere-dense}
A class of graphs $\mathscr{C}$ is said to be {\em nowhere-dense} if
for every $d \geq 0$ there exists a graph $H_d$ such that $H_d \not\preceq^d_m G$ for all $G \in \mathscr{C}$.
$\mathscr{C}$ is {\em effectively nowhere-dense} if the map $d \mapsto H_d$ is computable.
Otherwise, $\mathscr{C}$ is said to be {\em somewhere-dense}.
\end{definition}

Nowhere-dense classes of graphs were introduced by Nesetril and Ossona de Mendez~\cite{NO10,NO08}
and ``nowhere-density'' turns out to be a very robust concept with several
natural characterizations~\cite{GKS13}.
We use one such characterization in Section~\ref{subsec:nowheredense-is}.
It follows from the definition that planar graphs, graphs of bounded treewidth,
graphs of bounded degree, $H$-minor-free graphs, and
$H$-topological-minor-free graphs are nowhere-dense~\cite{NO10,NO08}.
As in the work of Dawar and Kreutzer~\cite{DK13}, we are only interested in effectively nowhere-dense classes;
all natural nowhere-dense classes are effectively nowhere-dense, but it is possible to construct
artificial classes that are nowhere-dense, but not effectively so.

\begin{definition}
A class of graphs $\mathscr{C}$ is said to be {\em $d$-degenerate} if there is an
integer $d$ such that every induced subgraph of any graph $G \in \mathscr{C}$ has a vertex of degree at most $d$.
\end{definition}

Graphs of bounded degeneracy and nowhere-dense graphs are incomparable~\cite{GKS14}.
In other words, graphs of bounded degeneracy are somewhere-dense.

\begin{proposition}[\cite{LW70}]\label{prop:degenerate-average}
The number of edges in a $d$-degenerate graph is at
most $dn$ and hence its average degree is at most $2d$.
\end{proposition}

Degeneracy is a hereditary property, hence any induced subgraph of
a $d$-degenerate graph is also $d$-degenerate.
It is well-known that graphs of treewidth at most $d$ are also $d$-degenerate.
Moreover a $d$-degenerate graph cannot
contain $K_{d+1,d+1}$ as a subgraph, which brings us to the
class of biclique-free graphs.
The relationship between bounded degeneracy, nowhere-dense,
and $K_{d,d}$-free graphs was shown by Philip et al. and Telle and Villanger~\cite{PRS09,TV12}.

\begin{definition}
A class of graphs $\mathscr{C}$ is said to be {\em $d$-biclique-free}, for some $d > 0$,
if $K_{d,d}$ is not a subgraph of any $G \in \mathscr{C}$, and it is
said to be {\em biclique-free} if it is $d$-biclique-free for some $d$.
\end{definition}

\begin{proposition}[\cite{PRS09,TV12}]
Any degenerate or nowhere-dense class of graphs is biclique-free,
but not vice-versa.
\end{proposition}

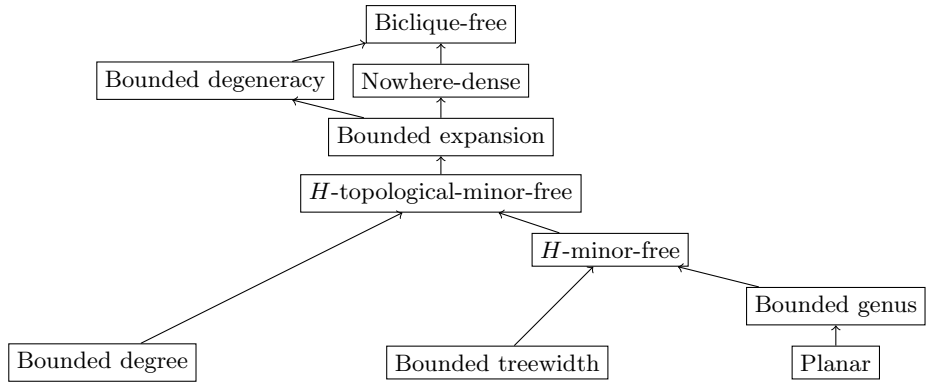
\begin{figure}
\centering
    \begin{tikzpicture}[->, scale=.75, auto=left, remember picture,every node/.style={rectangle},inner/.style={rectangle},outer/.style={rectangle}]

    %\draw [help lines] (0,0) grid (15,9);
    %\node[legendnoborder] at (1,6.5) {\small{$S$}};
    %\node[legendnoborder] at (13,6.5) {\small{$T$}};

    \node[outer] (DEGREE) at (0,0) [rectangle, fill=white, draw=black] {Bounded degree};
    \node[outer] (TREEWIDTH) at (7,0) [rectangle, fill=white, draw=black] {Bounded treewidth};
    \node[outer] (PLANAR) at (13,0) [rectangle, fill=white, draw=black] {Planar};
    \node[outer] (GENUS) at (13,1) [rectangle, fill=white, draw=black] {Bounded genus};
    \node[outer] (MINOR) at (9,2) [rectangle, fill=white, draw=black] {$H$-minor-free};
    \node[outer] (TOPMINOR) at (6,3) [rectangle, fill=white, draw=black] {$H$-topological-minor-free};
    \node[outer] (EXPANSION) at (6,4) [rectangle, fill=white, draw=black] {Bounded expansion};
    \node[outer] (DENSE) at (6,5) [rectangle, fill=white, draw=black] {Nowhere-dense};
    \node[outer] (DEGENERATE) at (2,5) [rectangle, fill=white, draw=black] {Bounded degeneracy};
    \node[outer] (BICLIQUE) at (6,6) [rectangle, fill=white, draw=black] {Biclique-free};

    \foreach \from/\to in {PLANAR/GENUS,TREEWIDTH/MINOR,GENUS/MINOR,MINOR/TOPMINOR,DEGREE/TOPMINOR}
    \draw (\from) -> (\to);

    \foreach \from/\to in {TOPMINOR/EXPANSION,EXPANSION/DENSE,DENSE/BICLIQUE,EXPANSION/DEGENERATE,DEGENERATE/BICLIQUE}
    \draw (\from) -> (\to);

    \end{tikzpicture}
\caption{Sparse graph classes~\cite{CLASSES,NO10,NO08}. Arrows indicate inclusion.}
\label{fig-graph-classes}
\end{figure}

\paragraph{Parameterized complexity.}
Using the framework developed by Downey and Fellows~\cite{DF97}, a
{\em parameterized problem} includes in the input a parameter $p$.
For a parameterized problem $\mathcal{Q}$ with inputs of the form $(x,p)$, $|x| = n$ and
$p$ a positive integer, $\mathcal{Q}$ is {\em fixed-parameter tractable}
(or in \FPT) if it can be decided in $f(p) n^c$ time,
where $f$ is an arbitrary function and $c$ is a
constant independent of both $n$ and $p$.
$\mathcal{Q}$ is in the class \XP\ if it can be decided in $n^{f(p)}$ time.
$\mathcal{Q}$ has a {\em kernel} of size $f(p)$ if there is an algorithm that
transforms the input $(x, p)$ to $(x', p')$ in
polynomial time (with respect to $|x|$ and $p$) such that $(x, p)$ is a
yes-instance if and only if $(x', p')$ is a
yes-instance, $p' \leq g(p)$, and $|x'| \leq f(p)$.
Each problem in \FPT\ has a kernel, possibly of exponential (or worse) size~\cite{DF97}.

In order to distinguish between parameterized problems solvable in $n^{f(p)}$ time and
parameterized problems solvable in $f(p)n^c$ time,
Downey and Fellows~\cite{DF97} introduced the {\em \WW-hierarchy}.
The hierarchy consists of a complexity class \WT\
for every integer $t \geq 1$ such that $\WT \subseteq \WTT$ for all $t$.
They proved that $\FPT \subseteq \WONE \subseteq \WTWO \subseteq \ldots \subseteq \WT$
and conjectured that strict containment holds.
In particular, the assumption $\FPT \subset \WONE$ is a natural
parameterized analogue of the conjecture that $\PP \neq \NP$.
Moreover, Downey and Fellows showed that the {\sc Independent Set} problem
parameterized by solution size is \WONE-complete
and the {\sc Dominating Set} problem parameterized by
solution size is \WTWO-complete.
Showing hardness in the parameterized setting is usually accomplished
using \FPT\ reductions.
The reader is referred to the books of
Niedermeier, Flum, and Grohe for more on parameterized complexity~\cite{FG06,N06}.

\paragraph{Reconfiguration.}
For any vertex-subset problem $\mathcal{Q}$, graph $G$, and positive integer $k$,
we consider the {\em reconfiguration graph} $R_{\mathcal{Q}}(G, k, k+1)$ when
$\mathcal{Q}$ is a minimization problem (e.g. {\sc Dominating Set})
and the reconfiguration graph $R_{\mathcal{Q}}(G, k - 1, k)$ when $\mathcal{Q}$ is
a maximization problem (e.g. {\sc Independent Set}).
A set $S \subseteq V(G)$ has a corresponding node in $V(R_{\mathcal{Q}}(G, r_l, r_u))$,
$r_l \in \{k - 1,k\}$ and $r_u \in \{k, k+1\}$, if and only if
$S$ is a feasible solution for $\mathcal{Q}$ and $r_l \leq |S| \leq r_u$.
We refer to {\em vertices} in $G$ using lower case letters (e.g. $u, v$) and to the {\em nodes} in
$R_{\mathcal{Q}}(G, r_l, r_u)$, and by extension their associated feasible solutions, using
upper case letters (e.g. $A, B$).
If $A,B \in V(R_{\mathcal{Q}}(G, r_l, r_u))$ then there exists an edge
between $A$ and $B$ in $R_{\mathcal{Q}}(G, r_l, r_u)$ if and only if
there exists a vertex $u \in V(G)$ such that $\{A \setminus B\} \cup \{B \setminus A\} = \{ u \}$.
Equivalently, for $A \Delta B = \{A \setminus B\} \cup \{B \setminus A\}$
the {\em symmetric difference} of $A$ and $B$, $A$ and $B$ share an edge in
$R_{\mathcal{Q}}(G, r_l, r_u)$ if and only if $|A \Delta B| = 1$.

We write $A \leftrightarrow B$ if there exists a path
in $R_{\mathcal{Q}}(G, r_l, r_u)$, a reconfiguration sequence, joining $A$ and $B$.
Any reconfiguration sequence from {\em source} feasible solution $S_s$ to
{\em target} feasible solution $S_t$, which we sometimes
denote by $\sigma = \langle S_0, S_1, \ldots, S_\ell \rangle$, for some $\ell$, has
the following properties:

\begin{itemize}
\item[-] $S_0 = S_s$ and $S_\ell = S_t$,
\item[-] $S_i$ is a feasible solution for $\mathcal{Q}$ for all $0 \leq i \leq \ell$,
\item[-] $|S_i \Delta S_{i+1}| = 1$ for all $0 \leq i < \ell$, and
\item[-] $r_l \leq |S_i| \leq r_u$ for all $0 \leq i \leq \ell$.
\end{itemize}

We denote the {\em length} of $\sigma$ by $|\sigma|$.
For $0 < i \leq \ell$, we say vertex $v \in V(G)$ is {\em added} at step/index/position/slot $i$
if $v \not\in S_{i - 1}$ and $v \in S_i$. Similarly, a vertex $v$
is {\em removed} at step/index/position/slot $i$ if $v \in S_{i - 1}$ and $v \not\in S_i$.
A vertex $v \in V(G)$ is {\em touched} in the course of a reconfiguration
sequence if $v$ is either added or removed at least once; it is {\em untouched} otherwise.
A vertex is {\em removable} ({\em addable}) from feasible solution $S$ if
$S \setminus \{v\}$ ($S \cup \{v\}$) is also a feasible solution for $\mathcal{Q}$.
For any pair of consecutive solutions ($S_{i - 1}$, $S_{i}$) in $\sigma$, we say
$S_{i}$ ($S_{i - 1}$) is the {\em successor} ({\em predecessor}) of $S_{i - 1}$ ($S_{i}$).
A reconfiguration sequence $\sigma' = \langle S_0, S_1, \ldots, S_{\ell'} \rangle$
is a {\em prefix} of $\sigma = \langle S_0, S_1, \ldots, S_{\ell} \rangle$ if $\ell' < \ell$.

We adapt the concept of irrelevant vertices from parameterized
complexity to introduce the notions of irrelevant and strongly irrelevant vertices for reconfiguration.
Since these notions apply to almost any reconfiguration problem, we give general definitions.

\begin{definition}\label{def-irr}
For any vertex-subset problem $\mathcal{Q}$, $n$-vertex graph $G$, positive integers $r_l$ and $r_u$,
and $S_s, S_t \in V(R_{\mathcal{Q}}(G, r_l, r_u))$ such that there exists a reconfiguration sequence from
$S_s$ to $S_t$ in $R_{\mathcal{Q}}(G, r_l, r_u)$, we say a vertex $v \in V(G)$ is {\em irrelevant} (with respect to $S_s$ and $S_t$)
if and only if $v \not\in S_s \cup S_t$ and there exists a reconfiguration sequence from $S_s$ to $S_t$
in $R_{\mathcal{Q}}(G, r_l, r_u)$ which does not touch $v$.
We say $v$ is {\em strongly irrelevant} (with respect to $S_s$ and $S_t$) if it is
irrelevant and the length of a shortest reconfiguration
sequence from $S_s$ to $S_t$ which does not touch $v$ is no greater than
the length of a shortest reconfiguration sequence which does (if the latter sequence exists).
\end{definition}

At a high level, it is enough to consider irrelevant vertices when trying to find {\em any}
reconfiguration sequence between two feasible solutions, but
strongly irrelevant vertices must be considered if we wish to find a {\em shortest}
reconfiguration sequence. As we shall see, our algorithm for {\sc DSR} does in fact find strongly irrelevant
vertices and can therefore be used to find shortest reconfiguration sequences. For {\sc ISR}, we
are only able to find irrelevant vertices and reconfiguration sequences are not guaranteed to be
of shortest possible length.

\section{Independent set reconfiguration}

\subsection{Graphs of bounded degeneracy}
To show that the {\sc ISR} problem is fixed-parameter tractable on
$d$-degenerate graphs, for some integer $d$,
we will proceed in two stages. In the first stage, we will show, for an instance $(G, I_s, I_t, k)$, that
as long as the number of low-degree vertices in $G$ is ``large enough'' we can
find an irrelevant vertex (Definition~\ref{def-irr}).
Once the number of low-degree vertices is bounded, a simple counting
argument (Proposition~\ref{fact:bound-on-high})
shows that the size of the remaining graph is also bounded and hence we can
solve the instance by exhaustive enumeration.

\begin{proposition}\label{fact:bound-on-high}
Let $G$ be an $n$-vertex $d$-degenerate graph, $S_1 \subseteq V(G)$ be the set of
vertices of degree at most $2d$, and $S_2 = V(G) \setminus S_1$.
If $|S_1| < s$, then $|V(G)| \leq (2d + 1)s$.
\end{proposition}

\begin{proof}
The number of edges in a $d$-degenerate graph is at
most $dn$ and hence its average degree is at most $2d$ (Proposition~\ref{prop:degenerate-average}).
If $|V(G)| = (2d + 1)s + c$, for $c \geq 1$, then $|S_2| = |V(G) \setminus S_1| > 2ds + c$, $\sum_{v \in S_2}{|N_G(v)|} > (2ds + c)(2d + 1)$, and
we obtain the following contradiction:
\begin{align}
{\sum_{v \in S_1}{|N_G(v)|} + \sum_{v \in S_2}{|N_G(v)|} \over |V(G)|} &> {(2ds + c)(2d + 1) \over (2d + 1)s + c} \notag \\
&= {4d^2s + 2ds + 2dc + c \over (2d + 1)s + c} \notag \\
&= {2d(2ds + s + c) + c \over 2ds + s + c} > 2d. \notag
\end{align}
\qed
\end{proof}

To find irrelevant vertices, we make use of the following classical result of Erd\~{o}s and Rado~\cite{ER60},
also known in the literature as the sunflower lemma. We first define the terminology
used in the statement of the theorem. A {\em sunflower} with $k$ {\em petals}
and a {\em core} $Y$ is a collection of sets $S_1,\ldots,S_k$ such that $S_i \cap S_j = Y$ for all $i \neq j$;
the sets $S_i \setminus Y$ are petals and we require none of them to be empty. Note that
a family of pairwise disjoint sets is a sunflower (with an empty core).

\begin{theorem}[Sunflower Lemma~\cite{ER60}]\label{theorem:sunflower}
Let $\mathscr{A}$ be a family of sets (without duplicates) over a universe $\mathscr{U}$,
such that each set in $\mathscr{A}$ has cardinality at most $d$.
If $|A| > d!(k - 1)^d$, then $\mathscr{A}$ contains a sunflower with $k$ petals and such a
sunflower can be computed in time polynomial in $|\mathscr{A}|$, $|\mathscr{U}|$, and $k$.
\end{theorem}

\begin{lemma}\label{lemma:low-degree}
Let $(G, I_s, I_t, k)$ be an instance of {\sc ISR} where $G$ is $d$-degenerate and
let $B$ be the set of vertices in $V(G) \setminus \{I_s \cup I_t\}$ of degree at most $2d$.
If $|B| > (2d + 1)!(2k - 1)^{2d + 1}$, then there exists an
irrelevant vertex $v \in V(G) \setminus \{I_s \cup I_t\}$ such that
$(G, I_s, I_t, k)$ is a yes-instance if and only if $(G', I_s, I_t, k)$ is a yes-instance,
where $G'$ is obtained from $G$ by deleting $v$ and all edges incident on $v$.
\end{lemma}

\begin{proof}
Let $b_1$, $b_2$, $\ldots$, $b_{|B|}$ denote the vertices in $B$
and let $\mathscr{A} = \{N_G[b_1]$, $N_G[b_2]$, $\ldots$, $N_G[b_{|B|}]\}$ denote the family of
sets corresponding to the closed neighborhoods of each vertex in $B$ and set $\mathscr{U} = \bigcup_{b \in B} N[b]$.
Since $|B|$ is greater than $(2d + 1)!(2k - 1)^{2d + 1}$, we know from Theorem~\ref{theorem:sunflower} that
$\mathscr{A}$ contains a sunflower with $2k$ petals and such a
sunflower can be computed in time polynomial in $|\mathscr{A}|$ and $k$.
Note that we assume, without loss of generality, that there are no two vertices $u$ and $v$
in $V(G) \setminus \{I_s \cup I_t\}$ such that $N_G[u] = N_G[v]$, as we
can safely delete one of them from the input graph otherwise, i.e. one of the two is (strongly) irrelevant.
Let $v_{ir}$ be a vertex whose closed neighborhood corresponds to one
of those $2k$ petals. We claim that $v_{ir}$ is irrelevant
and can therefore be deleted from $G$ to obtain $G'$.

To see why, consider any reconfiguration sequence
$\sigma = \langle I_s = I_0, I_1, \ldots, I_t = I_\ell \rangle$ from $I_s$ to $I_t$
in $R_{{\textsc{is}}}(G, k-1, k)$. Since $v_{ir} \not\in I_s \cup I_t$, we let
$p$, $0 < p < \ell$, be the first index in $\sigma$ at which $v_{ir}$ is added, i.e. $v_{ir} \in I_p$ and
$v_{ir} \not\in I_i$ for all $i < p$. Moreover, we let $q + 1$, $p < q + 1 \leq \ell$
be the first index after $p$ at which $v_{ir}$ is removed, i.e.
$v_{ir} \in I_{q}$ and $v_{ir} \not\in I_{q + 1}$. We will consider
the subsequence $\sigma_s = \langle I_p, \ldots, I_q \rangle$ and show
how to modify it so that it does not touch $v_{ir}$. Applying the same procedure
to every such subsequence in $\sigma$ suffices to prove the lemma.

Since the sunflower constructed to
obtain $v_{ir}$ has $2k$ petals and the size of any independent set in $\sigma$ (or
any reconfiguration sequence in general) is at most $k$, there must exist another {\em free}
vertex $v_{fr}$ whose closed neighborhood
corresponds to one of the remaining $2k - 1$ petals which we can add at index $p$ instead
of $v_{ir}$, i.e. $v_{fr} \not\in N_G[I_p]$. We say $v_{fr}$ {\em represents} $v_{ir}$.
Assume that no such vertex exists. Then we know that either some vertex in the
core of the sunflower is in $I_{p}$ contradicting the fact that
we are adding $v_{ir}$, or every petal
of the sunflower contains a vertex in $I_{p}$, which is not possible since
the size of any independent set is at most $k$ and the number of petals is larger.
Hence, we first modify the subsequence $\sigma_s$ by adding $v_{fr}$ instead
of $v_{ir}$. Formally, we have
$\sigma'_s = \langle (I_p \setminus \{v_{ir}\}) \cup \{v_{fr}\}, \ldots, (I_q \setminus \{v_{ir}\}) \cup \{v_{fr}\}\rangle$.

To be able to replace $\sigma_s$ by $\sigma'_s$ in $\sigma$ and obtain a reconfiguration
sequence from $I_s$ to $I_t$, then all of the following conditions must hold:
\begin{itemize}
\item[(1)] $|(I_q \setminus \{v_{ir}\}) \cup \{v_{fr}\}| = k$.
\item[(2)] $(I_i \setminus \{v_{ir}\}) \cup \{v_{fr}\}$ is an independent set of $G$ for all $p \leq i \leq q$,
\item[(3)] $|(I_i \setminus \{v_{ir}\}) \cup \{v_{fr}\} \Delta (I_{i+1} \setminus \{v_{ir}\}) \cup \{v_{fr}\}| = 1$ for all $p \leq i < q$, and
\item[(4)] $k - 1 \leq |(I_i \setminus \{v_{ir}\}) \cup \{v_{fr}\}| \leq k$ for all $p \leq i \leq q$.
\end{itemize}
It is not hard to see that if there exists no $i$, $p < i \leq q$,
such that $\sigma'_s$ adds a vertex in $N[v_{fr}]$ at position $i$, then
all four conditions hold. If there exists such a position, we
will modify $\sigma'_s$ into yet another subsequence $\sigma''_s$
by finding a new vertex to represent $v_{ir}$. The length of $\sigma''_s$
will be one greater than the length of $\sigma'_s$.

We let $i$, $p < i \leq q$, be the first position in $\sigma'_s$
at which a vertex in
$u \in N[v_{fr}]$ (possibly equal to $v_{fr}$) is added.
Using the same arguments discussed to find $v_{fr}$, and since
we constructed a sunflower with $2k$ petals, we can find another
vertex $v'_{fr}$ such that $N[v_{fr}] \cap I_{i - 1} = \emptyset$.
This new vertex will represent $v_{ir}$ instead of $v_{fr}$.
We construct $\sigma''_s$ from $\sigma'_s$ as follows:
$\sigma''_s = \langle I_p \setminus \{v_{ir}\} \cup \{v_{fr}\}, \ldots,
I_{i - 1} \setminus \{v_{ir}\} \cup \{v_{fr}\}, I_{i - 1} \setminus \{v_{ir}\} \cup \{v'_{fr}\}, I_{i} \setminus \{v_{ir}\} \cup \{v'_{fr}\},
\ldots, I_q \setminus \{v_{ir}\} \cup \{v'_{fr}\}\rangle$.
If $\sigma''_s$ now satisfies all four conditions then we are done.
Otherwise, we repeat the same process (which can occur at most $q - p$ times)
until we reach such a subsequence.
\qed
\end{proof}

\begin{theorem}\label{theorem-degenerate}
{\sc ISR} on $d$-degenerate graphs is fixed-parameter tractable
parameterized by $k + d$.
\end{theorem}

\begin{proof}
For an instance $(G, I_s, I_t, k)$ of {\sc ISR}, we know from Lemma~\ref{lemma:low-degree}
that as long as $V(G) \setminus \{I_s \cup I_t\}$ contains more than $(2d + 1)!(2k - 1)^{2d + 1}$ vertices
of degree at most $2d$ we can find an irrelevant vertex and reduce the size of the graph.
After exhaustively reducing the graph to obtain $G'$, we known that $G'[V(G') \setminus \{I_s \cup I_t\}]$,
which is also $d$-degenerate, has at most $(2d + 1)!(2k - 1)^{2d + 1}$ vertices of degree at most $2d$.
Hence, applying Proposition~\ref{fact:bound-on-high}, we know that
$|V(G') \setminus \{I_s \cup I_t\}| \leq (2d + 1)(2d + 1)!(2k - 1)^{2d + 1}$ and
$|V(G')| \leq (2d + 1)(2d + 1)!(2k - 1)^{2d + 1} + 2k$.
\qed
\end{proof}

\subsection{Nowhere-dense graphs}\label{subsec:nowheredense-is}
Nesetril and Ossona de Mendez~\cite{NM10} showed an interesting
relationship between nowhere-dense classes and a property of classes of structures
introduced by Dawar~\cite{DAWAR07,DAWAR10} called {\em quasi-wideness}.
We will use quasi-wideness and show a rather interesting relationship
between {\sc ISR} on graphs of bounded
degeneracy and nowhere-dense graphs. That is, our algorithm for nowhere-dense graphs
will closely mimic the previous algorithm in the following sense.
Instead of using the sunflower lemma to find a large sunflower, we will use
quasi-wideness to find a ``large enough almost sunflower'' with an initially ``unknown'' core
and then use structural properties of the graph to find this core and complete the sunflower.
We first state some of the results that we need.
Given a graph $G$, a set $S \subseteq V(G)$ is called {\em r-scattered}
if $N^r_G(u) \cap N^r_G(v) = \emptyset$ for all distinct $u, v \in S$.

\begin{proposition}\label{fact:scattered-flower}
Let $G$ be a graph and let $S = \{s_1, s_2, ..., s_k\} \subseteq V(G)$ be a
$2$-scattered set of size $k$ in $G$. Then the closed neighborhoods
of the vertices in $S$ form a sunflower with $k$ petals and an empty core.
\end{proposition}

\begin{definition}
A class $\mathscr{C}$ of graphs is {\em uniformly quasi-wide}
with {\em margin} $s_\mathscr{C} : \mathbb{N} \rightarrow \mathbb{N}$
and $N_\mathscr{C} : \mathbb{N} \times \mathbb{N} \rightarrow \mathbb{N}$ if for
all $r, k \in \mathbb{N}$, if $G \in \mathscr{C}$ and $W \subseteq V(G)$ with $|W| > N_\mathscr{C}(r, k)$,
then there is a set $S \subseteq W$ with $|S| < s_\mathscr{C}(r)$, such that $W$ contains
an $r$-scattered set of size at least $k$ in $G[V(G) \setminus S]$.
$\mathscr{C}$ is {\em effectively uniformly quasi-wide} if $s_\mathscr{C}(r)$ and
$N_\mathscr{C}(r, k)$ are computable.
\end{definition}

Examples of effectively uniformly quasi-wide classes include graphs
of bounded degree with margin $1$ and $H$-minor-free graphs with margin $|V(H)| - 1$.

\begin{theorem}[\cite{DK13}]
A class $\mathscr{C}$ of graphs is effectively
nowhere-dense if and only if $\mathscr{C}$ is effectively uniformly quasi-wide.
\end{theorem}

\begin{theorem}[\cite{DK13}]\label{theorem:quasi-wide}
Let $\mathscr{C}$ be an effectively nowhere-dense class of graphs and $h$ be the computable
function such that $K_{h(r)} \not\preceq^r_m G$ for all $G \in \mathscr{C}$.
Let $G$ be an $n$-vertex graph in $\mathscr{C}$, $r,k \in \mathbb{N}$, and $W \subseteq V(G)$ with $|W| \geq N(h(r),r,k)$,
for some computable function $N$. Then in $\Oh(n^2)$ time, we can compute a set
$B \subseteq V(G)$, $|B| \leq h(r) - 2$, and a set $A \subseteq W$ such that
$|A| \geq k$ and $A$ is an $r$-scattered set in $G[V(G) \setminus B]$.
\end{theorem}

\begin{lemma}\label{lemma:almost-flower}
Let $\mathscr{C}$ be an effectively nowhere-dense class of graphs and $h$ be the computable
function such that $K_{h(r)} \not\preceq^r_m G$ for all $G \in \mathscr{C}$.
Let $(G, I_s, I_t, k)$ be an instance of {\sc ISR} where $G \in \mathscr{C}$
and let $R$ be the set of vertices in $V(G) \setminus \{I_s \cup I_t\}$.
Moreover, let $\mathscr{P} = \{P_1, P_2, \ldots\}$ be a family of sets which partitions $R$
such that for any two distinct vertices $u,v \in R$, $u,v \in P_i$ if and only if
$N_{G}(u) \cap \{I_s \cup I_t\} = N_{G}(v) \cap \{I_s \cup I_t\}$.
If there exists a set $P_i \in \mathscr{P}$ such that $|P_i| > N(h(2),2, 2^{h(2) + 1} k)$, for some computable function $N$,
then there exists an irrelevant vertex $v \in V(G) \setminus \{I_s \cup I_t\}$ such that
$(G, I_s, I_t, k)$ is a yes-instance if and only if $(G', I_s, I_t, k)$ is a yes-instance,
where $G'$ is obtained from $G$ by deleting $v$ and all edges incident on $v$.
\end{lemma}

\begin{proof}
By construction, we known that the family $\mathscr{P}$ contains at most $4^k$ sets,
as we partition $R$ based on their neighborhoods in $I_s \cup I_t$. Note that some
vertices in $R$ have no neighbors in $I_s \cup I_t$ and will therefore belong to the same set in $\mathscr{P}$.

Assume that there exists a $P \in \mathscr{P}$ such that $|P| > N(h(2),2, 2^{h(2) + 1} k)$.
Consider the graph $G[R]$. By Theorem~\ref{theorem:quasi-wide}, we can, in $\Oh(|R|^2)$ time, compute a set
$B \subseteq R$, $|B| \leq h(2) - 2$, and a set $A \subseteq P$ such that
$|A| \geq 2^{h(2) + 1} k$ and $A$ is a $2$-scattered set in $G[R \setminus B]$.
Now let $\mathscr{P}' = \{P'_1, P'_2, \ldots\}$ be a family of sets which partitions $A$
such that for any two distinct vertices $u,v \in A$, $u,v \in P'_i$ if and only if
$N_{G}(u) \cap B = N_{G}(v) \cap B$. Since $|A| \geq 2^{h(2) + 1} k$ and $|\mathscr{P}'| \leq 2^{h(2)}$, we
know that at least one set in $\mathscr{P}'$ will contain at least $2k$ vertices of $A$.
Denote these $2k$ vertices by $A'$. All vertices in $A'$ have the same neighborhood
in $B$ and the same neighborhood in $I_s \cup I_t$ (as all vertices in $A'$ belonged to the same set $P \in \mathscr{P}$).
Moreover, $A'$ is a $2$-scattered set in $G[R \setminus B]$.
Hence, the sets $\{N_G[a'_1], N_G[a'_2], \ldots, N_G[a'_{2k}]\}$, i.e. the closed neighborhoods
of the vertices in $A'$,
form a sunflower with $2k$ petals (Proposition~\ref{fact:scattered-flower});
the core of this sunflower is contained in $B \cup I_s \cup I_t$.
Using the same arguments as we did in the proof of Lemma~\ref{lemma:low-degree}, we
can show that there exists at least one irrelevant vertex $v \in V(G) \setminus \{B \cup I_s \cup I_t\}$.
\qed
\end{proof}

\begin{theorem}
{\sc ISR} restricted to any effectively nowhere-dense class $\mathscr{C}$ of graphs
is fixed-parameter tractable parameterized by $k$.
\end{theorem}

\begin{proof}
If after partitioning $V(G) \setminus \{I_s \cup I_t\}$ into at most $4^k$ sets
the size of every set $P \in \mathscr{P}$ is bounded by $N(h(2),2, 2^{h(2) + 1} k)$, then
we can solve the problem by exhaustive enumeration, as $|V(G)| \leq 2k + 4^k N(h(2),2, 2^{h(2) + 1} k)$.
Otherwise, we can apply Lemma~\ref{lemma:almost-flower} and reduce the size of the graph in polynomial time.
\qed
\end{proof}

\section{Dominating set reconfiguration}

\subsection{\WONE-hardness}
The \WONE-hardness of the {\sc DSR} problem can be shown using only minor modifications
to the standard parameterized reduction from {\sc IS} to {\sc DS}. That is,
instead of reducing from {\sc IS} to {\sc DS}, we can instead give
a reduction from {\sc ISR} to {\sc DSR}.
We include a proof for completeness.

\begin{theorem}\label{theorem:dom-hard}
{\sc DSR} parameterized by $k$ is \WONE-hard on general graphs.
\end{theorem}

\begin{proof}
We let $(G, I_s, I_t, k)$ be an instance of {\sc ISR}, where $V(G) = \{v_1, \ldots, v_n\}$,
$E(G) = \{e_1, \ldots, e_m\}$, $I_s = \{v_{i_1}, \ldots, v_{i_k}\}$, and $I_t = \{v_{j_1}, \ldots, v_{j_k}\}$.
We first construct a graph $G'$ as follows.
$G'$ consists of the disjoint union of $k$ vertex-disjoint cliques
$C_1, \ldots, C_k$, each of size $n$, $k$ vertex-disjoint
independent sets $F_1, \ldots, F_k$, each of size at most $k + 2$, and at most
$n^2k^2$ vertex-disjoint independent sets $R_1, R_2, \ldots$, each of size $k + 2$.
Intuitively, each set $F_i$ will force any dominating set of $G'$ of size $k$ (or $k + 1$) to pick a vertex from each $C_i$
and the ``$R$ sets'' will guarantee that the selected vertices form an independent set in $G$.
Formally, we have:

\begin{itemize}
\item[(1)] For every vertex $v \in V(G)$ there is a corresponding vertex in
each $C_i$, $1 \leq i \leq k$ and we let $C_i = \{c^i_1, \ldots, c^i_n\}$.
\item[(2)] For every $1 \leq i \leq k$, we make the set $C_i$ a clique in $G'$.
\item[(3)] For each set $C_i$, $1 \leq i \leq k$, we introduce a set $F_i$ of $k + 2$ new
independent vertices and add an edge between each vertex in $C_i$ and all vertices in $F_i$.
\item[(4)] For a vertex $c^i_p \in C_i$ and a vertex $c^j_q \in C_j$, $i \neq j$, $1 \leq i,j \leq k$, and $1 \leq p,q \leq n$,
if $p = q$ or $v_pv_q \in E(G)$ we introduce $k + 2$ new
independent vertices and make them adjacent to all vertices
in $C_i \cup C_j \setminus \{c^i_p, c^j_q\}$. In other words, each new vertex
dominates all but two vertices
in $C_i \cup C_j$, namely $c^i_p$ and $c^j_q$.
\end{itemize}

We let $(G', D_s, D_t, k)$ denote the corresponding {\sc DSR} instance, where
$D_s = \{c^1_{i_1}, \ldots, c^k_{i_k}\}$ and $D_t = \{c^1_{j_1}, \ldots, c^k_{j_k}\}$.
Clearly, any dominating set $D$ of $G'$ of size $k$ must pick exactly one vertex
from each $C_i$, $1 \leq i \leq k$, and each such set corresponds to an independent set of
size $k$ in $G$. Moreover, any reconfiguration sequence between
$D_s$ and $D_t$ starts by adding a vertex (since $G'$ has no dominating set of size
$k - 1$) and then removing another (since dominating sets larger than $k + 1$ are
not allowed). By swapping the order of consecutive vertex additions and removals we obtain
a one-to-one correspondence between
reconfiguration sequences of independent sets of $G$ (of size $k$ and $k - 1$) and
reconfiguration sequences (of the same length) between
dominating sets of $G'$ (of size $k$ and $k+1$). The instances are thus equivalent.
\qed
\end{proof}

\subsection{Graphs excluding $K_{d,d}$ as a subgraph}
The parameterized complexity of the {\sc Dominating Set} problem (parameterized by $k$)
on various classes of graphs has been studied extensively in the literature;
the main goal has been to push the tractability frontier as
far as possible. The problem was shown fixed-parameter tractable
on planar graphs by Alber et al.~\cite{ABFN00}, on bounded genus graphs by Ellis et al.~\cite{EFF02},
on $H$-minor-free graphs by Demaine et al.~\cite{DFHT05}, on bounded
expansion graphs by Nesetril and Ossona de Mendez~\cite{NO08}, on nowhere-dense graphs
by Dawar and Kreutzer~\cite{DK13}, on degenerate graphs by Alon and Gutner~\cite{AG09},
and finally on $K_{d,d}$-free graphs by Philip et al.~\cite{PRS09} and Telle and Villanger~\cite{TV12}.
Figure~\ref{fig-graph-classes} illustrates the inclusion relationship among these classes of graphs,
which all fall under the category of sparse graphs.
Our fixed-parameter tractable algorithm relies on many of these earlier results.
Interestingly, and since the class of $K_{d,d}$-free graphs includes all those other graph classes,
our algorithm (Theorem~\ref{theorem:ds-fpt}) implies that the diameter of the reconfiguration graph $R_{{\textsc{ds}}}(G, k, k + 1)$
(or of its connected components), for $G$ in any of the aforementioned classes, is bounded above by $f(k,c)$,
where $f$ is a computable function and $c$ is constant which depends on the graph
class at hand. We start with some definitions and known results.

\begin{definition}[\cite{DDFKLPPRSVS14,SSPRIVATE14,PRS09,TV12}]\label{def:dom-core}
Given a graph $G$, the {\em domination core} of $G$ is a set $C \subseteq V(G)$
such that any set $D \subseteq V(G)$ is a dominating set of $G$ if and only if
$D$ dominates $C$. In other words, $D$ is a dominating set
of $G$ if and only if $C \subseteq N_G[D]$.
\end{definition}

\begin{theorem}[\cite{SSPRIVATE14,PRS09,TV12}]\label{theorem:domination-core}
If $G$ is a graph which excludes $K_{d,d}$ as a subgraph and $G$
has a dominating set of size at most $k$ then
the size of the domination core $C$ of $G$ is at most $dk^d$
and $C$ can be computed in $\Oh^*(dk^d)$ time.
\end{theorem}

\begin{definition}
A bipartite graph $G$ with bipartition $(A,B)$ is {\em $B$-twinless} if there are no
vertices $u,v \in B$ such that $N(u) = N(v)$.
\end{definition}

\begin{theorem}[\cite{SSPRIVATE14}]\label{theorem:twinless}
If $G$ is a bipartite graph with bipartition $(A,B)$ such that $G$ is $B$-twinless
and excludes $K_{d,d}$ as a subgraph then
\begin{eqnarray*}
|B| \leq 2(d - 1) ({|A|e \over d})^{2d}.
\end{eqnarray*}
\end{theorem}

Since Theorem~\ref{theorem:domination-core} implies a bound on the size of the domination
core and allows us to compute it efficiently, our main concern is to deal with
vertices outside of the core, i.e. vertices in $V(G) \setminus C$.
The next lemma shows that we can in fact find strongly irrelevant vertices outside
of the domination core of a graph.

\begin{lemma}\label{lemma-ds-irr}
For $G$ an $n$-vertex graph, $C$ the domination core of $G$, and $D_s$ and $D_t$
two dominating sets of $G$, if there exist $u,v \in V(G) \setminus \{C \cup D_s \cup D_t\}$
such that $N_G(u) \cap C = N_G(v) \cap C$ then $u$ (or $v$) is strongly irrelevant.
\end{lemma}

\begin{proof}
Given a reconfiguration sequence $\sigma = \langle D_0 = D_s, D_1, \ldots, D_\ell = D_t\rangle$
from $D_s$ to $D_t$ which touches $u$,
we will show how to obtain a reconfiguration sequence $\sigma'$ such that
$|\sigma'| \leq |\sigma|$ and $\sigma'$ touches $v$ but not $u$.

We construct $\sigma'$ in two stages.
In the first stage, we construct the sequence
$\alpha = \langle D'_0, D'_1, \ldots, D'_\ell \rangle$ of dominating sets,
where for all $0 \leq i \leq \ell$
\begin{align*}
   D'_i = \left\{
    \begin{array}{l l}
    \text{$D_i \cup \{v\} \setminus \{u\}$ if $u \in D_i$} \\
    \text{$D_i$ if $u \not\in D_i$.}
    \end{array} \right.
\end{align*}
Note that $\alpha$ is not necessarily a reconfiguration sequence
from $D_s$ to $D_t$. In the second stage, we repeatedly delete
from $\alpha$ any set $D'_i$ such that $D'_i = D'_{i + 1}$, $0 \leq i < \ell$.
We let $\sigma' = \langle D'_0, D'_1, \ldots, D'_{\ell'} \rangle$ denote the resulting sequence, in which there are
no two consecutive sets that are equal, and we claim
that $\sigma'$ is in fact a reconfiguration sequence from $D_s$ to $D_t$.

To prove the claim, we need to show that
the following conditions hold:
\begin{itemize}
\item[(1)] $D'_0 = D_s$ and $D'_{\ell'} = D_t$,
\item[(2)] $D'_i$ is a dominating set of $G$ for all $0 \leq i \leq \ell'$,
\item[(3)] $|D'_i \Delta D'_{i+1}| = 1$ for all $0 \leq i < \ell'$, and
\item[(4)] $k \leq |D'_i| \leq k + 1$ for all $0 \leq i \leq \ell'$.
\end{itemize}
Since $u,v \not\in D_s \cup D_t$, condition (1) clearly holds.
Moreover, since replacing $u$ by $v$ in any set does not increase
the size of the corresponding set, $k \leq |D'_i| \leq k + 1$ (condition (4) holds) and
$|D'_i \Delta D'_{i+1}| \leq 1$. As there are no two consecutive
sets in $\sigma'$ that are equal, $|D'_i \Delta D'_{i+1}| > 0$
and therefore $|D'_i \Delta D'_{i+1}| = 1$ (condition (3) holds).
The fact that $D'_i$ is a dominating set of $G$ follows from the
definition of a domination core.
Since $D_i$ is a dominating set of $G$, $C \subseteq N_G[D_i]$.
Moreover, since $N_G(u) \cap C = N_G(v) \cap C$ and $u,v \not\in C$,
we know that $C \subseteq N_G[D'_i]$.
By the definition of the domination core, it follows that $D'_i$ (which still dominates $C$)
is also a dominating set of $G$. Therefore, all four conditions hold, as needed.
\qed
\end{proof}

\begin{theorem}\label{theorem:ds-fpt}
{\sc DSR} parameterized by $k + d$ is fixed-parameter tractable on graphs that exclude $K_{d,d}$ as a subgraph.
\end{theorem}

\begin{proof}
Given a graph $G$, integer $k$, and two dominating sets $D_s$ and $D_t$ of $G$
of size at most $k$, we first compute the
domination core $C$ of $G$, which by Theorem~\ref{theorem:domination-core} can
be accomplished in $\Oh^*(dk^d)$ time.
Next, and due to Lemma~\ref{lemma-ds-irr}, we can delete all strongly irrelevant vertices
from $V(G) \setminus \{C \cup D_s \cup D_t\}$. We denote this new graph by $G'$.

Now consider the bipartite graph $G''$ with
bipartition $(A = C \setminus \{D_s \cup D_t\}, B = V(G') \setminus \{C \cup D_s \cup D_t\})$.
This graph is $B$-twinless, since for every pair of vertices $u,v \in V(G) \setminus \{C \cup D_s \cup D_t\}$
such that $N_G(u) \cap C = N_G(v) \cap C$ either $u$ or $v$ is strongly irrelevant
and is therefore not in $V(G')$ nor $V(G'')$. Moreover, since every subgraph of a $K_{d,d}$-free graph
is also $K_{d,d}$-free, $G''$ is $K_{d,d}$-free.
Hence, by Theorems~\ref{theorem:domination-core} and~\ref{theorem:twinless}, we have
\begin{eqnarray*}
|B| & \leq & 2(d - 1) ({|A|e \over d})^{2d} \\
& \leq & 2d(3|A|)^{2d} \leq 2d(3dk^d)^{2d}.
\end{eqnarray*}

Putting it all together, we know that after deleting all strongly
irrelevant vertices, the number of vertices in the resulting graph $G'$ is at most
\begin{eqnarray*}
|V(G')| & = & |V(C)| + |D_s \cup D_t| + |V(G') \setminus \{C \cup D_s \cup D_t\}|\\
& \leq & dk^d + 2k + 2d(3dk^d)^{2d}
\end{eqnarray*}

Hence, we can solve {\sc DSR} by exhaustively enumerating
all $2^{|V(G')|}$ subsets of $V(G')$ and building the
reconfiguration graph $R_{{\textsc{ds}}}(G',k,k+1)$.
\qed
\end{proof}

\bibliographystyle{abbrv}
\bibliography{references}

\end{document}